\documentclass[]{scrartcl}
    \usepackage[T1]{fontenc}
    \usepackage[utf8]{inputenc}
    \usepackage{amsthm,amssymb,amsmath}
    \usepackage[overload]{empheq}
    \usepackage{enumerate}
    \usepackage{color}
    \usepackage{multirow}
    \usepackage{tabularx}
    \newcolumntype{Y}{>{\centering\arraybackslash}X}
    \usepackage[svgnames,x11names]{xcolor}
    \usepackage{comment}
    \usepackage{soul}
    \usepackage{thmtools,thm-restate}
    
    \usepackage{subcaption}
    \usepackage{tikz}
    \usetikzlibrary{calc}
    
    \makeatletter
    \def\@maketitle{%
    \newpage%
    \null%
    \begin{center}%
        \let\footnote\thanks %
        {\huge \@title %
          \par
        }
      \vskip 1.5em
        {\lineskip .5em
         \begin{tabular}[t]{c}
            \baselineskip=12pt
            \@author
         \end{tabular}
         \par
        }
    \end{center}
    \par
    \vskip 1.5em}
    \makeatother
    \usepackage{hyperref}
    \hypersetup{
        colorlinks,
        linkcolor={red!50!black},
        citecolor={blue!50!black},
        urlcolor={blue!80!black}
    }
    \usepackage[%
        backend=biber,%
        style=ieee,%
        sortcites,%
        dashed=false,%
        citestyle=numeric,%
        sorting=nyt,%
        giveninits=true,%
        url=false,%
        doi=true,%
        isbn=false,%
        mincitenames=1,%
        maxcitenames=4,%
        maxbibnames=99,%
        defernumbers=true,%
        ]{biblatex}
    \usepackage[ocgcolorlinks]{ocgx2}%
    \usepackage[capitalise,nameinlink]{cleveref}
    \crefname{property}{Property}{Properties}
    \crefname{figure}{Figure}{Figures}
    
    \usepackage[defaultlines=3,all]{nowidow}
    
    \usepackage{microtype}
    \theoremstyle{plain}
    \newtheorem{theorem}{Theorem}[section]

    \newtheorem{definition}[theorem]{Definition}
    
    \newtheorem{lemma}[theorem]{Lemma}

    \theoremstyle{definition}
    \newtheorem{claim}{Claim}
    \newtheorem*{claim*}{Claim}
    
    \newenvironment{proofclaim}{
    	\noindent \emph{Proof.}
    }{\hfill $\diamond$ \\
    }

    \usepackage{tikzit}

\tikzstyle{SmallVtx}=[fill=black, draw=black, shape=circle, minimum size=4pt, inner sep=0pt]
\tikzstyle{SimpleVtx}=[fill=white, draw=black, shape=circle, minimum size=5pt, inner sep=0pt]

\tikzstyle{zone}=[-, fill=green, draw=green, opacity=.2]
\tikzstyle{zone2}=[-, fill=blue, draw=blue, opacity=.2, dashed]

    \usepackage{xspace}
    \def\eg{{\em e.g.}\@\xspace}
    
    \def\ie{{\em i.e.}\@\xspace}
    
    \newcommand{\NP}{\textsf{NP}}
    \renewcommand{\O}{\mathcal{O}}
    \renewcommand{\o}{{o}}
    \newcommand{\X}{\mathcal{X}}
    \newcommand{\XH}{\mathcal{X}_H}
    \def\GH{generalized house}\def\GB{generalized bull}
    \newcommand{\IMC}[1]{#1\textrm{-IMC}}

    \renewcommand{\leq}{\leqslant}
    \renewcommand{\geq}{\geqslant}
    \def\im{\subseteq_{im}}
    \def\nim{\nsubseteq_{im}}
    \def\FH{\widehat{K_4}}

    \newcommand{\IM}[2]{$#1 \subseteq_{im} #2$}
    
    \newcommand{\authorcite}[1]{\citeauthor{#1}~\cite{#1}}

    \usepackage[defaultlines=3,all]{nowidow}
    
    \usepackage[skins]{tcolorbox}
    \tcbuselibrary{many}
    \newtcolorbox{mypb}[2][]
    {
        enhanced,
        boxed title style = {colframe=white},
        attach boxed title to top left={
            xshift=0.5cm,
            yshift= -3.5mm,     
        },
        top=4mm,
        coltitle=black,
        beforeafter skip=\baselineskip,
        colframe = lightgray,
        colback  = white,
        colbacktitle  = white,
        coltitle = black,  
        fonttitle = \scshape,
        titlerule = 0mm, 
        title    = {#2},
        #1
    }
    
    \newcommand{\Pb}[4]{%
        \begin{mypb}{#1}
           \textbf{\textsf{Input}}: #2%
           \par\noindent%
           \textbf{\textsf{#4}}: #3?%
           \smallskip%
           \par\noindent%
        \end{mypb}
    }

    \title{Induced Minor Models. II. Sufficient conditions for polynomial-time detection of induced minors}
    
    \makeatletter
    \def\@maketitle{%
    \newpage%
    \null%
    \begin{center}%
        \let\footnote\thanks %
        {\huge \@title %
          \par
        }
      \vskip 1.5em
        {\lineskip .5em
         \begin{tabular}[t]{c}
            \baselineskip=12pt
            \@author
         \end{tabular}
         \par
        }
    \end{center}
    \par
    \vskip 1.5em}
    \makeatother
    
    \usepackage{authblk}

    \author[1]{Clément Dallard}
    \author[2,4]{Maël Dumas}
    \author[3]{Claire Hilaire}
    \author[4]{Anthony Perez}
    
    \newcommand{\email}[1]{
      \texttt{#1}
    }

    \affil[1]{Department of Informatics, University of Fribourg, Switzerland\protect\\\email{clement.dallard@unifr.ch}}
    \affil[2]{Institute of Informatics, University of Warsaw, Poland\protect\\
            \email{mael.dumas@mimuw.edu.pl}}
    \affil[3]{FAMNIT and IAM, University of Primorska, Slovenia\protect\\\email{claire.hilaire@upr.si}}
    \affil[4]{Université d'Orléans, INSA Centre Val de Loire, LIFO EA 4022, Orléans, France\protect\\\email{anthony.perez@univ-orleans.fr}}

    \makeatletter
    \def\blfootnote{\gdef\@thefnmark{}\@footnotetext}
    \makeatother

\begin{document}
    
    \maketitle%
        \blfootnote{This works is partly financed by the French \emph{Fédération de Recherche ICVL} (Informatique Centre-Val de Loire) and by the \emph{Slovenian Research and Innovation Agency} (research project J1-4008).}
    
    \begin{abstract}
        The \textsc{$H$-Induced Minor Containment} problem (\IMC{$H$}) consists in deciding if a fixed graph $H$ is an induced minor of a graph $G$ given as input, that is, whether $H$ can be obtained from $G$ by deleting vertices and contracting edges.
        Equivalently, the problem asks if there exists an induced minor model of $H$ in $G$, that is, a collection of disjoint subsets of vertices of $G$, each inducing a connected subgraph, such that contracting each subgraph into a single vertex results in $H$.

        It is known that \IMC{$H$} is \NP-complete for several graphs $H$, even when $H$ is a tree.
        In this work, we investigate which properties of $H$ guarantee the existence of an induced minor model whose structure can be leveraged to solve the problem in polynomial time.
        This allows us to identify four infinite families of graphs $H$ that enjoy such properties. 
        Moreover, we show that if the input graph $G$ excludes long induced paths, then $H$-IMC is polynomial-time solvable for any fixed graph $H$.
        As a byproduct of our results, this implies that \IMC{$H$} is polynomial-time solvable for all graphs $H$ with at most $5$ vertices, except for three open cases.

    \end{abstract}
    
    \section{Introduction}
    
    The notion of \emph{graph containment} has been intensively studied in the literature from both the algorithmic and structural viewpoints. %
    There are many ways to define whether a graph $G$ \emph{contains} a graph $H$, usually in terms of operations allowed on $G$ to obtain $H$.
    The most common operations are vertex deletion, edge deletion, and edge contraction.
    Any combination of these operations defines a \emph{graph containment relation}.
    The \emph{subgraph} relation only allows vertex and edge deletions, while the \emph{minor} relation also allows edge contractions.
    Their induced counterparts, namely the \emph{induced subgraph} and \emph{induced minor} relations, are defined analogously but without edge deletions. 
    These relations can then be used to define classes of graphs that \emph{exclude} a fixed collection $\mathcal{H}$ of graphs with respect to some fixed relation.
    Well-known graph classes can be characterized in terms of forbidden graphs.
    For example: weakly sparse graphs (sometimes simply referred to as sparse graphs) are those that exclude some fixed complete bipartite graph as a subgraph (see, \eg, \cite{Bonnet24,BONAMY2024215}); cographs are defined as graphs that exclude $P_4$ as an induced subgraph; planar graphs are characterized by excluding $K_5$ and $K_{3,3}$ as minors; graphs with bounded treewidth are those that exclude a planar graph as a minor; and chordal graphs correspond to graphs that exclude $C_4$ as an induced minor (see, \eg, \cite{BLS99}).
    A natural question that arises in this context is the complexity of determining whether a given graph $G = (V, E)$ \emph{contains} another graph $H$.
    If $H$ is part of the input, then the problem of determining if $H$ is a subgraph, induced subgraph, minor, or induced minor of a given graph $G$ is \NP-complete.\footnote{The \textsc{Hamiltonian Cycle} problem, with input graph $J$, can be reduced in polynomial-time to the considered problem, fixing $G$ to be the graph obtained from $J$ after subdividing every edge once and setting $H$ to be the $2|V(J)|$-vertex cycle.}
    However, if we consider $H$ as fixed, 
    then some of these problems can be solved in polynomial time.
    For the subgraph and induced subgraph relations, the corresponding problems can be solved in polynomial time by a simple brute-force approach, enumerating all (induced) subgraphs with $|V(H)|$ vertices.
    For the minor relation, there is the famous $\O(|V(G)|^3)$ algorithm by \authorcite{RS95}, later improved to $\O(|V(G)|^2)$ by \authorcite{KKR12}, and recently improved to almost linear time $\O((|V(G)|+|E(G)|)^{1+\o(1)})$ by \authorcite{KPS24}. 
    In sharp contrast, \authorcite{FKMP95} proved that when considering the induced minor relation, the problem is \NP-hard for some fixed graph $H$ on $68$ vertices. 
    It is thus natural to wonder for which graphs $H$ this problem is tractable.

    In this work, we focus on the \emph{induced minor} relation and consider, for several choices of $H$, the following problem: 
    
    \Pb{$H$-Induced Minor Containment (\IMC{$H$})}{A graph $G$.}{Does $G$ admit $H$ as an induced minor}{Question}

    This paper is the second of a series of paper started by a superset of the authors \cite{DDH+25} that focused on  structural properties of induced minor models, including bounds on treewidth and chromatic number of the subgraphs induced by minimal induced minor models. 
    This work complements the previous one by investigating which conditions on $H$ or $G$ are sufficient so that the problem becomes polynomial-time solvable.
    
    \paragraph{Related work} \authorcite{FKMP95} asked whether \IMC{$H$} can be solved in polynomial if $H$ is a tree or a planar graph. %
    Motivated by this question, \authorcite{FKM12} showed that \IMC{$H$} can be solved in polynomial time for all but one forest $H$ on at most $7$ vertices.
    The complexity of the remaining forest (obtained from two claws after identifying one of their leaves) is still open (and has recently been the subject of an open question at a workshop~\cite{CMPSA22}). 
    They also showed that, when $H$ is a subdivided star or obtained by adding at least two leaves to both of the endpoints of an edge, the problem remains polynomial-time solvable.
    Recently, \authorcite{KL24} eventually settled the complexity of the problem
    when $H$ is a tree and showed that there exists a tree (with over $2^{300}$ vertices) for which \IMC{$H$} is \NP-hard.
    In the same paper, the authors also gave a randomized polynomial-time algorithm that, given two graphs $G$ and $H$, outputs either an induced minor model of $H$ in $G$ or a balanced separator of $G$ of size $\O(\min(\log |V(G)|, |V(H)|^2) \cdot \sqrt{|V (H)| + |E(H)|} \cdot \sqrt{|E(G)|})$.
    In particular, if $H$ is fixed, this implies subexponential-time algorithms for several \NP-hard problems on $H$-induced-minor-free graphs.
    Previous to that, \authorcite{KORHONEN2023206} showed that graphs with large treewidth and bounded degree contain a large grid as an induced minor, which implies that, for every planar $H$, there is a subexponential time algorithm for \textsc{Max Weight Independent Set} on $H$-induced-minor-free graphs.
    
    Another recent result, by \authorcite{NSS24}, states that finding $s\geqslant 1$ pairwise anticomplete, disjoint cycles can be done in polynomial time.
    This in particular implies that if $H$ is a disjoint union of $s$ triangles, then \IMC{$H$} can be solved in polynomial time, thus answering the open question about the complexity of \IMC{$2C_3$} asked by \authorcite{FKM12}.
    
    Several results have also been obtained when restricting the input graph.
    For instance, \authorcite{FKMP95} showed that \IMC{$H$} is polynomial-time solvable, for any graph $H$, in the class of planar graphs (note that if $H$ is not planar, then the problem becomes trivial).
    Then, \authorcite{HKPST12} extended this result by showing that \IMC{$H$} can be solved efficiently on proper minor-closed graph classes,\footnote{A graph class $\mathcal{G}$ is \emph{minor-closed} whenever any minor of graph $G \in \mathcal{G}$ belongs to $\mathcal{G}$. It is \emph{proper} if it is not the class of all graphs.} %
    for any planar graph $H$.
    In similar fashion, \authorcite{GKP13} proved that the problem is polynomial-time solvable in AT-free graphs while \authorcite{BGHHKP11} showed the same result for chordal graphs. 
    
    Instead of considering graphs that \emph{exclude} some fixed graph $H$ as an induced minor, some works analyze the structure of graphs that \emph{contain} $H$ as an induced minor.
    For instance, \authorcite{CHKTW24} showed that if a graph contains $K_{3,4}$ as an induced minor, then it must contain a triangle or a theta\footnote{A \emph{theta} is a graph consisting of three internally, anticomplete, vertex disjoint paths that share the same two vertices as endpoints.} as an induced subgraph.

    \paragraph{Our results} 
    In the first paper of the series~\cite{DDH+24,DDH+25}, it is proved among other results  that if $H$ is the $4$-wheel, the $5$-vertex complete graph minus an edge, or a complete bipartite graph $K_{2,q}$, then there is a polynomial-time algorithm to solve \IMC{$H$}. %
    We carry on this line of research by settling the complexity status of \IMC{$H$} for all but three graphs with up to five vertices. 
    For each such graph $H$, we show that \IMC{$H$} can be solved in polynomial-time.
    Many cases actually follow from more general results, which we prove in this paper, that settle the complexity of \IMC{$H$} for some infinite classes of graphs.
    In particular, we obtain the following results.  Formal definitions of the family of graphs can be found in \cref{sec:BB_property}; see \cref{fig:order5list} for the list of such graphs with 5 vertices.  Informally, flowers are intersection of paths, cycles and diamonds in one vertex.

    \begin{restatable}{theorem}{ThmFlowers}\label{thm:flowerPoly}
        If $H$ is a flower, then \IMC{$H$} is polynomial-time solvable.
    \end{restatable}

    The generalized houses and bulls are obtained from houses and bulls by  subdividing the edges not in the triangle.
    
    \begin{restatable}{theorem}{ThmHousesAndBulls}\label{th:housePoly}
        If $H$ is a \GH{} or a \GB{}, then \IMC{$H$} is polynomial-time solvable.
    \end{restatable}
    
    The graph $S_{k,p}$ is the graph obtained by adding all edges between a clique of size $k$ and an independent set of size $p$.
    
    \begin{restatable}{theorem}{ThmCompleteSplits}\label{thm:splitPoly}
        Let $k \leq 3$ and $p$ be positive integers.
        Then \IMC{$S_{k,p}$} is polynomial-time solvable.
    \end{restatable}
    
    We emphasize that the class of flowers contains all subdivided stars. %
    Therefore, \cref{thm:flowerPoly} generalizes the result by \authorcite{FKM12}. %
    Let us also mention that \authorcite{MP24} independently showed that \IMC{House} can be solved in polynomial time (see \cref{fig:order5list} for a representation of the house).
    Their approach, different from ours, relies on an algorithm for detecting the house as an induced topological minor, and then reducing the induced minor case to the former. 
    In the same paper, the authors make use of our structural results for the flowers to detect butterflies (two triangles sharing one vertex) as an induced minor in polynomial time.
    
    \medskip
    
    When considering restricted input graphs, we broaden the complexity landscape by showing the following result on $P_t$-free graphs, that is, graphs without induced paths on $t$ vertices.
    We refer the reader to \cref{fig:order5list} for a representation of all graphs mentioned hereafter.
     
    \begin{restatable}{theorem}{ThmPtFree}\label{thm:ptfree}
        For any graph $H$ and any positive integer $t$, \IMC{$H$} is polynomial-time solvable in $P_t$-free graphs. 
    \end{restatable}
    
    \Cref{thm:ptfree} allows us to show that \IMC{Gem} and \IMC{$\FH$} are polynomial-time solvable (see \cref{thm:gem,thm:fullhouse}). 

    Note that $Gem$-induced-minor-free graphs and $\FH$-induced-minor-free graphs may contain arbitrarily long induced paths. 
    Nonetheless, leveraging their structure, we show that \IMC{$Gem$} and \IMC{$\FH$} on general graphs can be reduced to graphs without long induced paths.
    Thus, the two problems are polynomial-time solvable as a consequence of \cref{thm:ptfree} (see \cref{thm:gem,th:fhPoly}). Note that the polynomial-time solvability of \IMC{$Gem$} and \IMC{$\FH$} also follows from the fact that $Gem$-induced-minor-free graphs and $\FH$-induced-minor-free graphs have bounded clique-width (by results of \authorcite{BOS18}) and that \IMC{$H$} can be solved in polynomial time on graphs of bounded clique-width~\cite{CMR00,FK24}.

    \section{Preliminaries}\label{sec:preliminaries}

    We consider simple, undirected graphs $G = (V,E)$, where $V$ denotes the \emph{vertex set} and $E$ the \emph{edge set}.
    We may also use $V(G)$ to denote the vertex set of $G$ and $E(G)$ its edge set to clarify the context.  
    Given a vertex $u \in V$, the \emph{open neighborhood} of $u$ is the set %
    $N_G(u) = \{v \in V:\ uv \in E\}$. 
    The \emph{closed neighborhood} of $u$ is defined as $N_G[u] = N_G(u) \cup \{u\}$. 
    Given a subset of vertices $S \subseteq V$, $N_G[S]$\ is the set $\cup_{v \in S} N_G[v] $
    and $N_G(S)$ is the set $N_G[S] \setminus S$. 
    We will omit the mention to $G$ whenever the context is clear. 
    Given a set of vertices $S \subseteq V$, the subgraph of $G$ \emph{induced by} $S$, denoted $G[S]$, is the graph $(S,E_S)$ where $E_S = \{uv \colon \{u,v\} \in S \times S$\}.  
    In a slight abuse of notation, we use $G \setminus S$ to denote the graph induced $G[V \setminus S]$.
    Given two sets of vertices $A,B \subseteq V$, we say that $A$ and $B$ are 
    \emph{adjacent} if there exist $u \in A$ and $v \in B$ such that $uv$ is an edge of $G$.
    
    Given an edge $uv$ of $G$, we define the \emph{contraction} of $uv$ as the graph 
    obtained from $G$ by removing $u$ and $v$ and by adding a new vertex $w$ with neighborhood 
    $N(\{u,v\})$.
    Similarly, the \emph{subdivision} of $uv$ is obtained by 
    removing the edge $uv$ from $E$ and inserting a new vertex $w$ and edges $wu, wv$. 
    
    We may denote a path $P$ with $\ell$ vertices by a sequence $p_1 \dots p_\ell$ of vertices such that two consecutive vertices in the sequence are adjacent. 
    The path on $\ell$ vertices is denoted by $P_\ell$.
    The vertices $\{p_2, \ldots, p_{\ell-1}\}$ are called the \emph{internal vertices} of $P_\ell$.
    Similarly, a sequence $p_1 \dots p_\ell p_1$ describes a cycle $C$ with $\ell$ vertices such that two consecutive vertices in the sequence are adjacent.
    The cycle on $\ell$ vertices is denoted by $C_\ell$. 
    The edges of a path, or of a cycle, are the edges between consecutive vertices of the sequence and the \emph{length} of a path, or cycle, is the number of edges it has. 
    Given a path $P$ and some vertices $u,v$ of $P$, we let $uPv$ be the subpath of $P$ with extremities $u$ and $v$. If $w$ is adjacent to $u$, then $wuPv$ is the path obtained by adding the edge $wu$ to $uPv$, and similarly, if $w$ is adjacent to $v$, then $uPvw$ is the path obtained by adding the edge $vw$ to $uPv$. 
    
    \paragraph{Induced minor models} A graph $H$ is an \emph{induced minor} of $G$, 
    denoted \IM{H}{G}, whenever $H$ can be obtained from $G$ by removing vertices and contracting edges. 
    An \emph{induced minor model} of $H$ in $G$, or simply a \emph{model} of $H$, is a collection $\XH = \{X_u \colon u \in V(H)\}$ of pairwise disjoint non-empty subsets of $V(G)$ such that:
    \begin{itemize}
        \item for $u \in V(H)$, $G[X_u]$ is connected, and
        \item for $u \neq v \in V(H)$, $X_u$ and $X_v$ are adjacent if and only if $uv \in E(H)$.
    \end{itemize}
    
    Each set $X_u \in \XH$ is called a \emph{bag} of $\XH$. 
    The subgraph of $G$ \emph{induced by} $\XH$ is the subgraph induced by the union of the bags of $\XH$. 
    We say that a bag $X_u$ is trivial if $|X_u|=1$.
    
    A model $\XH'$ of $H$ is said to be \emph{included} in another model $\XH$ of $H$ if the union of the bags of $\XH'$ is included in the union of the bags of $\XH$. Note that it is not required that each bag of $\XH'$ is a subset of a bag of $\XH$. Given $S\subset V(H)$, we say that $\XH$ \emph{minimizes the size of the bags of $S$} (or just \emph{minimizes the bags of $S$}) if there is no model $\XH' = \{X'_u \colon u \in V(H)\}$ of $H$ included in $\XH$ such that $\sum_{v\in S} |X'_v| < \sum_{v\in S} |X_v|$. 
    In particular, we say that $\XH$ is a \emph{minimal model of $H$} if $\XH$ minimizes the bags of $V(H)$. 
    Finally, we say that a bag $X_u$ of $\XH$ is \emph{minimal} if there is no strict subset $X_u'$ of $X_u$ such that replacing $X_u$ by $X_u'$ results in a model of $H$. Note in particular that, if $\XH$ is a minimal model of $H$, then each bag of $\XH$ is minimal.

    A \emph{premodel} is a collection of disjoint subset of vertices of $G$, $\X = \{X_u \colon u \in V(H)\}$, that is not necessarily a model of $H$. 
    In particular, $X_u$ can be the empty set.
    We say that a model $\X^* = \{X^*_u \colon u \in V(H)\}$ of $H$ in $G$ \emph{extends} a premodel $\X = \{X_u \colon u \in V(H)\}$ if, for each $u\in V(H)$, we have $X_u \subseteq X_u^*$.

    \medskip 
    
    Finally, note that given a graph $G$, a graph $H$, and a collection of pairwise disjoint subsets $\X = \{X_u \colon u \in V(H)\}$ of $V(G)$, deciding if $\X$ is a model of $H$ in $G$ can be done in time $\O(|V(G)|^2)$. Indeed, it is enough to check that $G[X_u]$ is connected for each $X_u \in \X$ and that for $u,v \in V(H)$, $uv \in E(H)$ if and only if there is $xy \in E(G)$ with $x\in X_u$ and $y\in X_v$.

    \subsection{Graphs with at most 5 vertices}\label{sec:order5} 
    Before diving into our more general proofs, we provide \cref{fig:order5list} an exhaustive list of graphs with $5$ vertices and recall known and new results regarding the complexity of \IMC{$H$}, for various graphs $H$. 
    \begin{figure}
        \centering
        \footnotesize
        \input{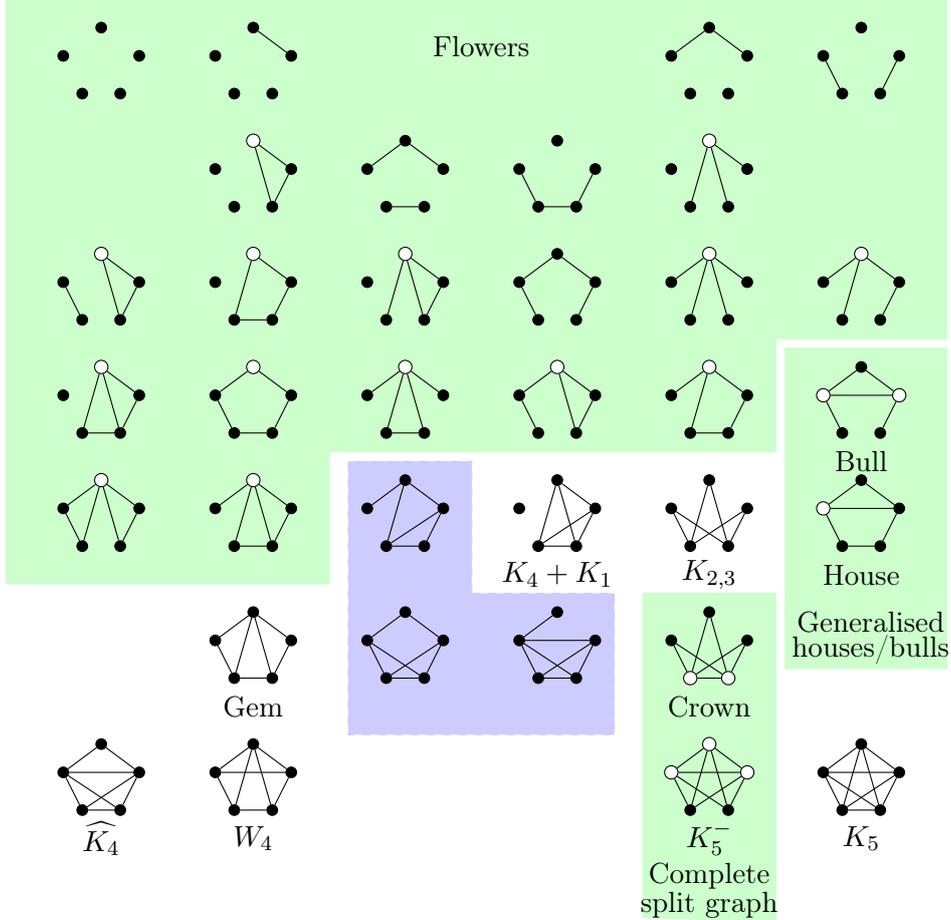}
        \caption{Exhaustive list of graphs with $5$ vertices. The group of graphs with green background belongs to infinite families studied in this paper.
        The ones with blue background are the ones for which the complexity of \IMC{$H$} remains open.
        }
        \label{fig:order5list}
    \end{figure}
    We first note that whenever $H$ is a clique, having $H$ as an induced minor is equivalent to having $H$ as a minor.
    Hence, \IMC{$H$} is polynomial-time solvable whenever $H$ is a clique~\cite{RS95,KKR12,KPS24}. 
    A similar observation can be made for \IMC{$K_t + K_1$}: for each choice of vertex $x$ of $G$ as a bag for $K_1$, all that remains is to try to find a model of $K_t$ in $G\setminus N[x]$.
    Moreover, it is easily noticed that a graph contains a path $P$ as an induced minor if and only if it contains $P$ as an induced subgraph.
    More generally, if $H$ is a disjoint union of paths, then a graph contains $H$ as an induced minor if and only if it contains $H$ as an induced subgraph; see \cref{lem:0BB}. 
    Therefore, in such a case, \IMC{$H$} can be solved in polynomial time.
    For graphs $H$ with at most four vertices, a short discussion in the introduction of~\cite{DDH+24} explains that determining if $H$ is an induced minor can be done efficiently.

    For $K_{2,3}$, a recent result by a superset of the authors of the current paper, based on a characterization in terms of forbidden induced subgraphs and the so-called shortest-path detectors technique, lead to a polynomial-time algorithm for \IMC{$K_{2,3}$}~\cite{DDH+24}. The same (super)set of authors, in an ongoing project, was able to extend these results to show that, among others, \IMC{$W_4$} is polynomial-time solvable. 
    
    \medskip All other results can be deduced from our work. In particular, we show \cref{sec:BB_property} that \IMC{$H$} is polynomial-time solvable whenever $H$ is a flower (\cref{sec:flower}), the bull or the house (\cref{sec:houses}), and the two complete split graphs (\cref{sec:split}). Let us mention once again that our results 
    apply to some generalizations of aforementioned graphs, and thus imply polynomial-time 
    algorithms for graphs with more vertices. Moreover, the result for flowers encompass known 
    results of \authorcite{FKM12} who proved polynomial-time solvability of \IMC{$H$} whenever 
    $H$ is a subdivided star.

    The cases where $H$ is the \emph{Full House} (denoted also $\FH$, which is a $K_4$ plus a vertex adjacent to two vertices of that $K_4$) or the \emph{Gem} (a $P_4$ plus a vertex complete to it) are discussed in \cref{sec:Pt_free}.
    
    \medskip

    \section{Almost trivial models} \label{sec:BB_property}
    
    Note that when considering induced minor models of $H$ in $G$, we can focus on models with as few vertices as possible, and thus on models that minimizes the bags of $V(H)$ (recall that we minimize over all the model, not each bag individually).
    In particular, if this size amounts to $\vert V(H) \vert$, then  the model induces a subgraph in $G$ isomorphic to $H$. 
    
    Based on this observation, we can restate the induced subgraph relation by saying that for every graph $G$ admitting $H$ as an induced subgraph there exists a model of $H$ in $G$ such that every bag is trivial (\ie contains exactly one vertex). In this section, we consider a natural generalization of this observation, where we allow only a subset of vertices of $H$ to have non-trivial bags.
    
    \begin{definition}[$S$-non-trivial property]
        \label{def:bb}
        Let $H$ be a graph and $S$ be a (potentially empty) set of vertices of $H$. We say that $H$ is \emph{$S$-non-trivial}, 
        or \emph{$S$-NT} for short, if for every graph $G$ admitting $H$ as an induced minor, there exists a model $\XH = \{X_u \colon u \in V(H)\}$ of $H$ in $G$ such that for each $v \in V(H) \setminus S$, $\vert X_v \vert = 1$. 
    \end{definition}
    
    Observe in particular that if $H$ is $\emptyset$-non-trivial, then $H$ is an induced minor of some graph $G$ if and only if $H$ is an induced subgraph of $G$.  Recall that, in this case, the problem can be trivially solved in polynomial time. 
    
    \medskip
    
    In the remaining of this section, we prove that \IMC{$H$} is polynomial-time solvable for $S$-NT graphs $H$ with $\vert S \vert \leqslant 1$.  
    We first give some structural properties on the bags of a model for vertices of small degree in $H$. The properties in \cref{lem:minModelSmallDeg} are already known and used  in other papers~\cite{FKM12,DDH+24}, for the sake of completeness, we prove these results. 
    A consequence of the following \lcnamecref{lem:minModelSmallDeg} is that paths are $\emptyset$-NT. 
    
    \begin{lemma}\label{lem:minModelSmallDeg} 
        Let $G$ and $H$ be two graphs such that \IM{H}{G}. 
        Let $\XH$ be a model of $H$ in $G$, such that $X_u$ is minimal for a vertex $u\in V(H)$. Then: 
        \begin{itemize}
            \item if $\deg_H(u)\leq 1$, then $|X_u|=1$; 
            \item if $\deg_H(u)= 2$, with neighbors $v,w$, 
            then there is a unique vertex $x_v$ in $N(X_v)\cap X_u$
             and a unique vertex $x_w$ in $N(X_w)\cap X_u$,
            and $X_u$ induces in $G$ a path whose extremities are $x_v,x_w$.
        \end{itemize}
    \end{lemma}
    \begin{proof}
        We define $H,G,\XH,u,X_u$ as in the statement of the lemma. 
        Suppose first that $\deg_H(u)\leq 1$, and that $|X_u|>1$. 
        Let $x_u$ be an arbitrary vertex of $X_u$ if $\deg_H(u)=0$, otherwise $x_u$ is an arbitrary vertex of $X_u$ adjacent to $X_v$ where $v$ is the unique neighbor of $u$. 
        Then we can replace $X_u$ by $\{x_u\}\subset X_u$ in $\XH$, which contradicts the minimality of $X_u$.

        Suppose now that $u$ has degree $2$ with neighbors $v,w$. Let $P$ be a shortest path in $X_u$ from a neighbor of $X_v$ to a neighbor of $X_w$: such a path exists since $X_u$ is connected and adjacent to $X_v$ and $X_w$. We denote the extremities of $P$ respectively $x_v,x_w$. Then there is no vertex in $V(P)\setminus\{x_v\}$ adjacent to $X_v$ since that would imply the existence of a path with the same property as $P$ that is shorter than $P$, and similarly, there is no vertex in $V(P)\setminus\{x_w\}$ adjacent to $X_w$. Moreover, since $P$ is a shortest path in $X_u$ it is in particular an induced path in $G$.
        Thus $X_u=V(P)$,  otherwise replacing $X_u$ by $V(P) \subset X_u$ in 
        $\XH$ would yield a model of $H$ contradicting the minimality of $X_u$.
    \end{proof}

    The following result guarantees that, if a graph admits some graph $H$ as an induced minor, then there exists a minimal model $\XH$ of $H$ such that  
    for every connected component of $H$ that is not a cycle, every vertex of $H$ of degree at most $2$ has a trivial bag in $\XH$. 
    Moreover, if a connected component of $H$ is a cycle, then at most one bag of this cycle is non-trivial in $\XH$.

    \begin{lemma}\label{lem:minModelPath}
        Let $H$ be a graph and $P$ be a path in $H$ such that the internal vertices of $P$ have degree $2$ (potentially the extremities of $P$ can be adjacent). 
        Let $G$ be a graph such that \IM{H}{G}, and $\XH$ a model of $H$ in $G$. 
        Then there is a model of $H$ in $G$ included in $\XH$ 
        such that the internal vertices of $P$ have trivial bags. Moreover, only the bag of one of the extremities of $P$ is bigger than it was in $\XH$, and the bags of $H\setminus V(P)$ are the same as in $\XH$.
    \end{lemma}
    \begin{proof}
        We define $H,G,\XH$ as in the lemma. 
        Let $P$ be a path in $H$ of extremities $a,b$ such that the internal vertices of $P$ have degree two. 
        We can assume that the bags of $V(P)$ are minimal up to removing unnecessary vertices in those bags, which can only make bags included in the original ones.
        By \cref{lem:minModelSmallDeg}, for every internal vertex $u$ of $P$, $X_u$ induces a path on at least 1 vertex between the two bags of the neighbors of $u$.
        Hence, the union of the bags of the internal vertices of $P$ induces a path $Q$ of extremities $y_a,y_b$, the unique neighbors of respectively $X_a$ and $X_b$ in $\bigcup_{v\in V(P)\setminus \{a,b\}} X_v=V(Q)$, and $Q$ contains at least $|V(P)|-2$ vertices.
    
        Then we can define a new model of $H$ by replacing in $\XH$ the bags of the internal vertices of $P$ respectively by the $|V(P)|-2$ first vertices of $Q$ (starting from $x_a$), and adding the remaining vertices of $Q$ to $X_b$. 
        Then this new model is included in $\XH$, the bags of $(V(H)\setminus V(P)) \cup \{a\}$ are the same as in $\XH$, 
        and all the internal vertices of $P$ have trivial bags.
    \end{proof}

    From the above lemmas, we can deduce the following result.
    
    \begin{lemma}\label{lem:0BB}
        $H$ is $\emptyset$-NT if and only if $H$ is a disjoint union of paths.
    \end{lemma}
    \begin{proof}
        Let $H$ be a $\emptyset$-NT graph. Recall that, in that case, $H$ is an induced minor of some graph $G$ if and only if $H$ is an induced subgraph of $G$.
        Suppose that $H$ has a cycle $C$. Let $G$ be the graph constructed from $H$ by subdividing an edge of $C$. Observe that $H$ is an induced minor of $G$ but not an induced subgraph of $G$, a contradiction.
        Suppose now that $H$ has a vertex $u$ with $\deg_H(u)\geq 3$. Let $G$ be the graph constructed from $H$ by replacing $u$ by a clique $K_u$ of size $\deg_H(u)$, and adding a matching between the neighbors of $u$ and the vertices of $K_u$. 
        Observe that $H$ is an induced minor of $G$, but there is no induced acyclic subgraph of $G$ of size $|V(H)|$, hence $H$ is not an induced subgraph of $G$, a contradiction.  
        Therefore, $H$ is acyclic and has maximum degree 2, and thus $H$ is a union of paths.
    
        Conversely, let $H$ be a disjoint union of paths, and $G$ be a graph that admits $H$ as an induced minor. 
        We get from \cref{lem:minModelPath} that there is a model $\XH$ such that the internal vertices of the paths of $H$ have trivial bags, and, up to removing useless vertices in each bag, we can suppose that $\XH$ has minimal bags.
        From \cref{lem:minModelSmallDeg}, we obtain that the extremities of those paths also have trivial bags. Thus, $H$ is $\emptyset$-NT. 
    \end{proof}
    
    \citeauthor{FKM12} observed that subdivided stars of center $u$ are $\{u\}$-NT, and gave a polynomial time algorithm for detecting them~\cite[Proposition 2]{FKM12}. We generalize their result for every $\{u\}$-NT graph $H$.
    
    \begin{theorem}\label{th:1BBpoly}
        If $H$ is $S$-NT with $|S|\leq 1$, then \IMC{$H$} is polynomial-time solvable.
    \end{theorem}
    \begin{proof} If $S = \emptyset$, then by \cref{lem:0BB} it is equivalent to testing if $H$ is a disjoint union of paths, which can clearly be done in polynomial time. Suppose that $H$ is $\{u\}$-NT for one of its vertex $u$. 
        Let $G$ be the input graph and assume that \IM{H}{G}. Since $H$ is $\{u\}$-NT, then there is a model $\XH$ where only the bag of $u$ is non-trivial. Observe that $G[X_u]$ induces a connected graph. Moreover, for each vertex $v$ not adjacent to $u$ in $H$, $X_u$ is not adjacent to $X_v$, so $X_u\cap N(X_v)= \emptyset$, and similarly, $X_u$ has to contain at least one vertex of $N(X_w)$ for each $w$ adjacent to $u$. 
    
        This gives us the following polynomial strategy to detect if $H$ is an induced minor of $G$, and output a model in the positive case:
        we enumerate all the premodels of $H$ where the bags contain exactly one vertex of $G$, except for the bag of $u$ which is empty. There are $\O(n^{|V(H)|-1})$ possibilities.
        Given such a premodel $\X=\{X_v \colon w \in V(H)\setminus \{u\}\}$, we first check if for each $v,w$ in $H\setminus \{u\}$, the vertices in the trivial bags $X_v$ and $X_w$ are adjacent if and only if $v,w$ are adjacent. If this condition is not satisfied, we can reject the premodel.
        Otherwise, let $Y = \bigcup_{v\in N_H(u)} X_v$ and $Z = \bigcup_{v\notin N_H(u)} X_v$. We enumerate the connected components $C_1,\dots ,C_r$ of  $G \setminus(Y\cup N[Z])$, which can be done by Breadth-First Search in time $\O(|V(G)|+|E(G)|)$. If there is one connected component $C_i$ containing a vertex of $N_G(v)$ for each $v \in Y$, then we have found a model of $H$ in $G$ with $X_u = C_i$.
        If for every possible premodel we did not find a suitable connected component, then we can conclude that $H$ is not an induced minor of $G$.
        The algorithm described here takes polynomial time, since $H$ is fixed.  
    \end{proof}

    \subsection{Flowers}\label{sec:flower}
    
    We say that a graph $H$ is a \emph{flower} if there is a vertex $u\in V(H)$ such that $H\setminus \{u\}$ is a disjoint union of paths and for each path $P$, either $|V(P)|=3$ and $P$ is complete to $u$ (\emph{sepal}), or $P$ is connected only by $0$, $1$ (\emph{stamens}) or $2$ (\emph{petal}) of its extremities to $u$. 
    The vertex $u$ is called the \emph{center} of $H$. We refer the reader to \cref{fig:order5list} for an exhaustive list of flowers with $5$ vertices. 
    
    \medskip
    
    We show that flowers are $\{u\}$-NT and hence can be detected in polynomial time with \cref{th:1BBpoly}. 
    
    \begin{lemma}\label{lem:flower1BB}
        If $H$ is a flower of center $u$, then $H$ is $\{u\}$-NT.
    \end{lemma}
    \begin{proof}
        Suppose $H$ is a flower of center $u$, let $G$ be a graph such that \IM{H}{G} and let
        $\XH$ be a model of $H$ in $G$ that minimizes the size of the bags of $H\setminus \{u\}$ (\ie such that there is no model $\XH'$ included in $\XH$ such that the sum of the sizes of the bags of $H\setminus \{u\}$ is strictly smaller in $\XH'$ than in $\XH$). In particular, every bag of $H\setminus \{u\}$ is minimal. 
        Let us show that $X_u$ is the only bag that is non-trivial.
    
        Suppose first that $H$ contains a sepal, \ie a path $P=abc$ that is complete to $u$, and suppose that the bag of a vertex of $P$ is not trivial.
        Let $y_a,y_u,y_c$ be three vertices adjacent to $X_b$, respectively belonging to $X_a,X_u,X_c$.
        Let $P_{ac}$ be a shortest path in $G[X_b\cup\{y_a,y_c\}]$ from $y_a$ to $y_c$. Note that $y_a$ and  $y_c$ are not adjacent and thus $P_{ac}$ admits at least one internal vertex in $X_b$.
        Let $P_u$ be a shortest path in $G[X_b\cup \{y_u\}]$ from $y_u$ to some internal vertex of $P_{ac}$.
        Then the vertex $x_b$ at the intersection of $P_{ac}$ and $P_u$ has degree at least $3$ in $G[V(P_{ac})\cup V(P_u)]$ and belongs to $X_b$. 
        Let $x_u$ be the neighbor of $x_b$ on $P_u y_u$, and similarly let $x_a$ be the neighbor of $x_b$ on $x_b P_{ac} y_a$ and $x_c$ be the neighbor of $x_b$ on $x_b P_{ac} y_c$. 
        Note that it is possible to have $x_a=y_a$, $x_b=y_b$ or $x_u=y_u$, but by construction, $\{x_a,x_b,x_c,x_u\}$ are all distinct.
        Similarly, our construction allows $x_u$ being adjacent to $x_a$ or $x_c$, but the fact that $P_{ac}$ is a shortest path in $G[X_b\cup\{y_a,y_c\}]$ prevents $x_a$ and $x_c$ from being adjacent. 
        Therefore, if we replace in $\XH$ the bags $X_a,X_b,X_c$ and $X_u$ by respectively $\{x_a\},\{x_b\},\{x_c\}$ and $X_u'=X_u\cup X_a\cup X_b\cup X_c \setminus \{x_a,x_b,x_c\}$ (in particular $X_u'$ contains $x_u$ that is adjacent to $x_b$), we obtain a new model of $H$ included in $\XH$ %
        in which the bags of $a,b$ and $c$ are trivial, and this contradicts the choice of $\XH$. 
    
        We showed that all the vertices of $H$ that are in a sepal (except $u$) have trivial bags.
        Observe that every vertex $v\neq u$ that is in a stamen or petal either has degree $1$ or is an internal vertex of a path with degree $2$.  
        Thus by \cref{lem:minModelSmallDeg,lem:minModelPath}, every such vertex has a trivial bag in $\XH$. 
        Hence, every vertex that is not $u$ has a trivial bag in $\XH$, therefore $H$ is $\{u\}$-NT.
    \end{proof}
    
    Combining \cref{lem:flower1BB} and \cref{th:1BBpoly}, we obtain the following \lcnamecref{thm:flowerPoly}. 
    
    \ThmFlowers*

    \subsection{Generalized Houses and Bulls}
    \label{sec:houses}
    
    We say that a graph is a \emph{\GH{}} if it consists of a triangle $a,u,v$, vertices $b$ and $c$ adjacent to $u$ and $v$ respectively, and a path  $R = b_1b_2\dots b_r$  from $b=b_1$ to $c=b_r$, with $r > 1$ (see \cref{fig:generalized_house}).
    A \GB{} is defined similarly where $R$ has a missing edge. More formally, in a \GB{} $u,v,a,b,c$ and their adjacencies are defined the same, but $R$ is replaced by two paths $R_b= b_1\dots b_s$ and $R_c=b_{s+1}\dots b_r$, with $1\leq s < r $, with still  $b=b_1$ and $c=b_r$. 
    
    We show that if $H$ is a \GH{} or a \GB{}, then \IMC{$H$} can be solved in polynomial time.
    The main idea here is that these graphs are $\{u,v\}$-NT, with one of the bags having a specific structure (\cref{lem:2BBhouses}).
    It also allows us to prove that \GH{s} are $\{u\}$-NT (\cref{lem:1BBhouses}). However, it is not the case for \GB{s}, and we design a polynomial-time algorithm that works for both families, outputting a model where both the bags of $u$ and $v$ might be non-trivial.
    
    \begin{figure}[ht]
        \centering
        \begin{subfigure}[t]{.48\textwidth}
            \centering\includegraphics[scale=2]{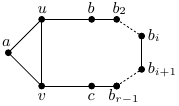}
            \caption{A generalized house}
            \label{fig:genhouse}
        \end{subfigure}
        \begin{subfigure}[t]{.48\textwidth}
            \centering\includegraphics[scale=2]{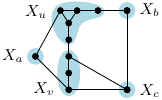}
            \caption{A model of a house}
            \label{fig:modelhouse}
        \end{subfigure}
        \caption{Note that removing one edge $b_ib_{i+1}$ for some $1\leq i \leq r-1$ results in a generalized bull. 
        }
        \label{fig:generalized_house}
    \end{figure}

    The idea of the proof of \cref{lem:2BBhouses} is similar to that of \cref{lem:flower1BB}: we start from a model that minimizes the bag of $v$ and deduce the structure of the bags.

    \begin{lemma}\label{lem:2BBhouses}
        If $H$ is a \GH{} or a \GB, then $H$ is $\{u,v\}$-NT. 
        Moreover, if a graph $G$ admits $H$ as induced minor, then there exists a model $\XH$ such that
         $G[X_v]$ is a path from a vertex adjacent to both $X_a$ and $X_c$ to a vertex adjacent to both $X_c$ and $X_u$. Furthermore, these vertices are the unique vertices in $X_v$ adjacent to $X_a$ and $X_u$ respectively. 
    \end{lemma}
    \begin{proof}
        Let $G$ be a graph that admits $H$ as an induced minor.
        First, observe that in $H$, $a$ and the internal vertices of $R$  (or $R_b,R_c$ for the generalized bull case) are all either internal vertices of a path with degree 2 or have degree 1. 
        Hence, by \cref{lem:minModelPath,lem:minModelSmallDeg}, we can always find a model of $H$ in $G$ such that the  
        corresponding bags are trivial.
        Thus, $H$ is $\{u,v\}$-NT. 
        We now prove that we can find a model where the bag of $v$ has the desired structure. 
        Let $\XH$ be a model of $H$ in $G$ with only $X_u$ and $X_v$ as non-trivial bags, that minimizes the size of $X_v$. 
        Notice in particular that the bag of each vertex of $H$ except $u$ is minimal. For a trivial bag $X_w, w\in V(H)$, let $x_w$ be the only vertex of $X_w$. 
        
        Let $P_{ac}$ be a shortest path in $X_v$ from a neighbor $y_a$ of $x_a$ to a neighbor $y_c$ of $x_c$, and let $P_{u}$ be a shortest path in $X_v$ from a neighbor of $X_u$ to a vertex of $P_{ac}$. If there are several choices for $P_u$, we choose one such that the intersection $x$ of $P_u$ and $P_{ac}$ is the closest to $y_c$ along $P_{ac}$.
        By connectivity, such paths exist, and by minimality of $X_v$, we have that $X_v=V(P_{ac})\cup V(P_u)$. 
        We want to show that $P_{ac}$ is reduced to a single vertex, so we suppose first that $P_{ac}$ has length at least $1$. 
        
        Let $P_a=y_aP_{ac}x$ and $P_c=xP_{ac}y_c$. Observe that no vertex of $P_a\setminus \{x\}$ is adjacent to $x_c$, 
        and in $P_c$, only $x$ can be adjacent to $x_a$ or $X_u$, and only $y_c$ can be adjacent to $x_c$. 
        Suppose that $P_a$ has length at least $1$, and let $x'_a$ be the neighbor of $x$ in $P_a$. 
        Then the collection $\X' = \{X'_w \colon w \in V(H)\}$ defined as: 
        \begin{align*}[left=\empheqlbrace]
                X_a' & = \{x_a'\}\\
                X_v' & = V(P_c)\cup V(P_u)\\
                X_u' & = X_u\cup\{x_a\}\cup(V(P_a)\setminus\{x_a',x\}) \\
                X_w' &= X_w\ \textstyle{for\ any\ other\ } w \in V(H)
        \end{align*}
        is a model of $H$ included in $\XH$ with a smaller bag for $v$, a contradiction.
        Therefore $x=y_a$ and $P_c=P_{ac}$. 
        
        Suppose now that $P_{c}$ has length at least 1, and let $x'_c$ be the neighbor of $x$ in $P_c$.
        Let us first assume that $H$ is a \GH{}. 
        Observe that $P= x'_cP_{c}x_{c}x_{b_{r-1}} \dots x_{b_2}x_{b}$ is  a path in $G$ 
        such that only $x_b$ is adjacent to $X_u$ and no vertex of $P$ is adjacent to $x_a$.
        Then we can define a new model $\XH' = \{X'_w \colon w \in V(H)\}$ of $H$ such that 
        $X'_a=X_a$, $X'_v=X_v\setminus (V(P_c)\setminus\{x\})$, $X'_c=\{x'_c\}$, 
        $X'_{b_{r-1}}$ is the next vertex of $P$ and so on for all the vertices of $R$. 
        Finally, $X'_u$ is obtained by adding the remaining vertices of $P$ to $X_u$, ensuring the adjacency with $X'_b$. 
        If $H$ is a \GB{}, we proceed similarly with $P= x'_cP_{c}x_{c} \dots x_{b_{s+1}}$; 
        then $X'_v$ and the bags from $X'_c$ to $X'_{b_{s+1}}$ are defined similarly, and $X_w'=X_w$ for any other $w\in V(H)$.
        It is then easy to check that in both cases, $\XH'$ is a model of $H$ included in $\XH$ with a smaller bag for $v$, a contradiction.
        
        Therefore, $P_{ac}$ is reduced to a single vertex $x$ that is adjacent to both $x_a$ and $x_c$, 
        and $X_v$ induces the path $P_u$ from $x$ to a neighbor $y_u$ of $X_u$. By minimality of $X_v$, $x$ is the only vertex of $X_v$ adjacent to both $x_a$ and $x_c$, and $y_u$ is the only vertex of $X_v$ adjacent to $X_u$. Observe that if $y_u$ is adjacent to $x_c$, then $x$ is the only vertex adjacent to $x_a$, otherwise we could restrict $P_u$ starting from another vertex of $P_u$ adjacent to $x_a$ that is closer to $y_u$ than $x$, which would result in a smaller bag for $v$.
        It only remains to prove that $y_u$ is adjacent to $x_c$.
        Suppose that $y_u$ is not adjacent to $x_c$. Then the collection $\XH' = \{X'_w \colon w \in V(H)\}$ with $ X_v' = X_v\setminus \{y_u\}$, 
        $X_u'=X_u\cup\{y_u\}$ and $X_w'=X_w$ for any other $w\in V(H)$ is a model of $H$ included in $\XH$ with a smaller bag for $v$, a contradiction. 

    \end{proof}

    From a model of the generalized house with $X_u$ and $X_v$ non-trivial, we can construct a model with only $u$ having a non-trivial bag (illustrated in \cref{fig:GH_1BB_model}). 
    
        \begin{lemma}\label{lem:1BBhouses}
            If $H$ is a \GH{}, then $H$ is $\{u\}$-NT.
        \end{lemma}
        \begin{proof}
             Let $G$ be a graph such that \IM{H}{G} and $\XH$ a model of $H$ as described in \cref{lem:2BBhouses} (keeping the same notation).
             Then we can construct a model of $H$ in $G$ such that only $u$ has a non-trivial bag (see \cref{fig:GH_1BB_model} for the case of the house).
             We can assume that $X_v$ contains at least two vertices, otherwise $\XH$ is already the sought model, 
             and that the vertex in $X_v$ adjacent to both $X_u$ and $X_c$ is not adjacent to $X_a$, by minimality of $X_v$. 
             Let $x_1$ be the vertex of $X_v$ adjacent to both $X_a$ and $X_c$ (recall that this vertex exists and is unique by \cref{lem:2BBhouses}). 
             Let $x_t$ be a vertex of $X_u$ adjacent to $X_b$. 
             Let $x_1\dots x_t$ be a shortest path in $G[X_u\cup X_v]$ from $x_1$ to $x_t$. 
             Observe that this path contains at least $3$ vertices, otherwise $X_v$ would contain only one vertex. 
             
             Then the collection $\XH' = \{X'_w \colon w \in V(H)\}$ defined as:
             \begin{align*}[left=\empheqlbrace]
                 X'_a &= \{x_1\}\\
                 X'_u &= \{x_i \colon 1<i<t\}\\
                 X'_{b_1} &= \{x_t\}\\
                 X'_{b_i} &= X_{b_{i-1}}, 2 \leqslant i \leqslant r
             \end{align*} 
             and $X'_v=X_{b_r}$ (recall that $b=b_1$ and $c=b_r$) is a model of $H$.
        \end{proof}

        Unfortunately, this construction does not extend to generalized bulls, as some of them are not $\{u\}$-NT, see \cref{fig:GH_1BB_model}. However, \cref{th:housePoly} presents a polynomial time algorithm for detecting \GB{s}, constructing models with possibly two non-trivial bags. 
        The idea behind the algorithm is, given $H$, to compute first a premodel of $H$ where each bag contains one vertex, and such that the bags are adjacent if and only if the vertices are adjacent in $H$, except between the bags $X_u,X_v$ and $X_u,X_b$ that might not be adjacent yet. Then, if this premodel can be extended into a model of $H$, we show that we can connect $X_u, X_b$ by choosing an arbitrary path between their respective vertices, then connecting the vertex of $X_v$ to this path. 
    
        \begin{figure}[ht]
            \centering
            \begin{subfigure}[t]{0.48\linewidth}
                \centering\includegraphics[scale=2]{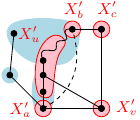}
                \caption{~\label{fig:GH1}}
            \end{subfigure} 
            \begin{subfigure}[t]{0.48\linewidth}
                \centering\includegraphics[scale=2]{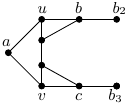}
                \caption{~\label{fig:GH2}}
            \end{subfigure}
            \caption{(a) Construction of a model for the house with only one non-trivial bag. (b) Example of a graph that admits the bull with subdivided horns as induced minor, but with at least two big bags in \emph{every} model.}
            \label{fig:GH_1BB_model}
        \end{figure}
        
    \ThmHousesAndBulls*
    
    \begin{proof}
    
    If $H$ is a \GH{}, the result follows directly from \cref{lem:1BBhouses} and \cref{th:1BBpoly}. In the following we present a method to solve \IMC{$H$} in polynomial time for \GB{s}. However, this method works the same way for \GH{s}.

    Let $G$ be a graph. From \cref{lem:2BBhouses} we know that if $G$ admits $H$ as an induced minor, then there exists a model of $H$ 
    where only the bags of $u$ and $v$ are non-trivial, and such that the bag of $v$ contains a vertex adjacent to both the bags of $a$ and $c$. Therefore,  
    there is a subgraph $Q$ in $G$ that forms a couple of paths $x_b\dots x_{b_{s}}$ and $x_{b_{s+1}}\dots x_cx_vx_ax_u$, with $x_w$ belonging to the bag of $w\in V(H)$ in that model (in the case of the \GH{}, $Q$ must form a path). 
    From each such subgraph $Q$ in $G$, we will try to construct a model of $H$. We start with a premodel $\X$ with $X_w = \{x_w\}$ for $w\in V(H)$.
    At this step, the bag $X_u$ may not be adjacent to $X_v$ and $X_b$, but for every other pair of vertices $w,w'$ in $H$, $X_{w'}$ and $X_w$ must be adjacent if and only if $w$ and $w'$ are adjacent. If this is not the case, we can reject the premodel. 
    
    We will now add vertices to the bags $X_u$ and $X_v$ in order for $\X$ to become a model of $H$. 
    Observe that if a model extending $\X$ exists, then the bags of $u$ and $v$ do not contain any neighbor of $x_{b_i}, 1<i<r$, so we can restrict the graph to $G'=G\setminus (\bigcup_{i=2}^{r-1} N[x_{b_i}]\setminus  \{x_b,x_c\}])$ (we keep $x_b$ and $x_c$ in $G'$ for the sake of simplicity).

    If there is a connected component $C$ in $G'\setminus (N[x_c] \cup \{x_a,x_b,x_v\})$ adjacent to $x_a,x_b$ and $x_v$, then we can set $X_u = C$ and we have constructed a model of $H$.
    Suppose now that there is no such component. 
    Observe that if there is a model of $H$ that extends $\X$, then, there exists a path $P_{ub}$ from $x_u$ to $x_b$ in $G'\setminus(\{x_a,x_v\}\cup N[x_c])$, and a path $P_v$ from $x_v$ to a vertex of $P_{ub}$ in $G'\setminus(\{x_a,x_b,x_c\})$. 
    Moreover, if there is a model of $H$ that extends $\X$ and we have not already found one, this means that for every pair of paths $P_{ub}$ and $P_v$, the latter must intersect the neighborhood of $x_c$. In particular, this means that the bag of $v$ must be non-trivial. 
    
    In the following claim, we show that to construct a model of $H$ extending $\X$, the choice of the path $P_{ub}$ does not matter. 
    
    \begin{claim}
        If there exists a model of $H$ extending $\X$ with only $u,v$ having non-trivial bags, then for any path $P_{ub}$ from $x_u$ to $x_b$ in $G'\setminus(\{x_a,x_v\}\cup N[x_c])$, there exists a model $\X'$ extending $\X$ with the internal vertices of $P_{ub}$ in $X'_u$.
    \end{claim}
    
    \begin{proofclaim}
        Let $\X' = \{X'_w \colon w \in V(H)\}$ be a model of $H$ that extends $\X$ 
        and that first minimizes the size of $X_v'$ and then the size of $X_u'$. 
        We can assume in particular that all the bags in $\X'$ are minimal. 
        We can moreover assume that $X_v'$ contains only the vertices of a path $P_v'$ from $x_v$ to a vertex $y_{uc}$ adjacent to both $x_c$ and a vertex of $X'_u$, otherwise it would not be minimal (note that this path avoids the neighbors of $x_b$).    
        
        Now, let $P_{ub}$ be any path from $x_u$ to $x_b$ in $G'\setminus N[x_c]$. We want to add vertices to $X_u,X_v$ to create a new model of $H$ extending $\X$ and such that the internal vertices of $P_{ub}$ belong to $X_u$.
        First, we add the internal vertices of $P_{ub}$ into $X_u$.
        If $P_{ub}$ and $P'_v$ does not intersect, then let $P'_{vu}$ be a path in $G[X'_u\cup\{y_{uc}\}]$ from $y_{uc}$ to $x_u$. Then adding $P'_{vu}$ into $X_u$ and $P'_v$ into $X_v$ result in a model of $H$. 
        Else, $P_{ub}$ and $P'_v$ intersect. Let $z_v$ be the vertex on both paths that is the closest to $x_v$ in $P'_v$. Observe that the subpath of $P'_v$ going from $x_v$ to $z_v$ does not contain any vertex of $P_{ub}$ except $z_v$. Hence, adding only the internal vertices of this subpath into $X_v$ results in a model of $H$.    
    \end{proofclaim}

    Let $P_{ub}$ be a path from $x_u$ to $x_b$ in $G'\setminus(\{x_a,x_v\}\cup N[c])$. Put its internal vertices into $X_u$.
    
    \begin{claim}
        If there exists a model of $H$ extending $\X$, then there exists $z\in N(x_c)$ such that there exists a model $\X' = \{X'_w \colon w \in V(H)\}$ extending $\X$ with $X'_v$ containing the vertices of any shortest path from $x_v$ to $z$ in $G'\setminus (N[x_b]\cup\{x_a,x_c\})$.
    \end{claim}
    \begin{proofclaim}
        Let $\X' = \{X'_w \colon w \in V(H)\}$ be a model of $H$ that extends $\X$ and minimizes the size of $X_v'$. 
        We can assume in particular that all the bags in $\X'$ are minimal.
        We can moreover assume that $G[X_v']$ is a path from $x_v$ to a vertex $z\in X_v'$ adjacent to both $x_c$ and a vertex in $X_u'$, otherwise $X_v'$ would not be minimal. 
        Now consider a shortest path $P_v'$ from $x_v$ to $z$ in $G'\setminus (N[x_b]\cup\{x_a,x_c\})$. If replacing $X_v'$ by the vertices of the path $P_v'$ does not yield a model, this means that $P'_v$ intersects $X_u'$. 
        Then we consider the subpath of $P_v'$ starting in $x_v$ and ending in the first vertex $z'$ adjacent to a vertex of $X_u'$:
        replacing $X_v'$ by the vertices of that subpath yield a model of $H$.
        Note that this is a shortest path from $x_v$ to $z$ in $G'\setminus (N[x_b]\cup\{x_a,x_c\})$, and this path contains a neighbor of $x_c$ (different from $x_v$).    
    \end{proofclaim}
    
    For each vertex $y_c \in N(x_c)$, we try to find a shortest path $P_v$ from $x_v$ to $y_c$ in $G'\setminus (N[x_b]\cup\{x_a,x_c\})$, and add its vertices to $X_v$. We next try to find a path from a neighbor of $y_c$ to any vertex on the path $P_{ub}$ (except $x_b$) in the graph $G\setminus (V(P_v)\cup N(x_c))$ and add its vertices into $X_u$. If both paths are found, then we constructed the sought bags of $u$ and $v$ and thus found a model of $H$ in $G$. If for each $y_c \in N(x_c)$, no pair of paths are found, then there was no model extending $\X$. 
    
    We repeat this process for each possible couple of paths $Q$. If $G$ admits $H$ as induced minor, this process will find eventually a model of $H$. As $H$ is fixed, this process can be done in polynomial time. 
    \end{proof}

    \subsection{Complete split graphs}\label{sec:split}
    
    Let $k,p \in \mathbb{N}$.
    The graph $S_{k,p}$, obtained by adding all possible edges between a clique of size $k$ and an independent set of size $p$, is called a \emph{complete split graph}.
    For $k=2$ and $p=3$, $S_{k,p}$ is also known as the \emph{Crown}, and for $k=3$ and $p=2$, it corresponds to $K_5^-$ ($K_5$ minus an edge); see \cref{fig:order5list} for a graphical representation.
    In this section, we show that if $k\leq3$, then \IMC{$S_{k,p}$} can be solved in polynomial time. The idea of the algorithm is to first guess the $p$ vertices in the independent set. Then we try to guess $k$ pairwise disjoint sets, each containing at least one neighbor of each of the $p$ vertices, and check if we can construct a model of a clique using these sets. The last part can be done in quasi-linear time with the algorithm of \authorcite{KPS24} for the \textsc{Rooted Minor Containment problem}. In this problem, given graphs $G$ and $H$, and a premodel (a \emph{root}) of $H$ in $G$, the goal is to find a model of $H$ in $G$ that extends the given premodel. 
    In the theorem below, the $O_{H,|X|}(\cdot)$-notation hides factors that depend on $H$ and $X$, and are computable. The set $X$ is the set of vertices in the root.
    \begin{theorem}[{\cite[Theorem 1.1]{KPS24}}]\label{thm:rooted_minor}
        The \textsc{Rooted Minor Containment} problem can be solved in time $O_{H,|X|}((|V(G)|+|E(G)|)^{1+\o(1)})$. In case of a positive answer, the algorithm also provides a model within the same running time. 
    \end{theorem}

    Using this result, we can prove the following theorem. We say that a clique $K$ of a graph $G$ is \emph{universal} if for every $x\in K$, $N[x] = V(G)$.
    
    \begin{theorem}\label{thm:DBBpoly}
        Let $H$ be a graph with some universal clique $K$.
        If $H$ is $K$-NT, then \IMC{$H$} is polynomial-time solvable.
    \end{theorem}
    \begin{proof} 
        Suppose that the graph $H$ is $K$-NT for some universal clique $K$ of size $k$.
        Let $G$ be a graph and $I = V(H)\setminus K$. Since $H$ is $K$-NT, if $G$ admits $H$ as an induced minor, then there exists a model of $H$ in $G$ where the bags of the vertices of $I$ are trivial. Moreover, observe that for each $u\in K$, the bag of $u$ is a subset of $V(G)$ that induces a connected subgraph of $G$ that contains at least one neighbor of the vertex in the bag of each $v \in V(H)\setminus K$. 
        
        This gives us the following polynomial strategy to detect if $H$ is an induced minor of $G$, and output a model in the positive case:
        We enumerate all the premodels of $H$ where the bags of vertices of $I$ contain exactly one vertex of $G$ and the bags of vertices of $K$ are empty. There are $\O(n^{|I|})$ such premodels.
        Given such a premodel $\X = \{X_v : v\in V(H)\}$, we can check first that for each $v,w\in I$, the vertices in the trivial bags $X_v$ and $X_w$ are not adjacent in $G$. 
        If this condition is not satisfied, we can reject the premodel. 
        
        Next, we try to construct $k$ pairwise disjoint subsets $Z_1,\dots, Z_k$ of size at most $|I|$ of $V(G) \setminus(\bigcup_{X\in \X} X)$ such that 
        for each $i \in [k]$ and each $v\in I$, we have $|Z_i\cap N(X_v)| = 1$. 
        Observe that there are $\O(n^{k|I|})$ such sets. Moreover, observe that there might be no such set of subsets, in this case, we can reject the current premodel.
        
        Then, for each possible $Z_1,\dots, Z_k$, we use \cref{thm:rooted_minor} to determine in polynomial time if $G \setminus(\bigcup_{X\in \X} X)$ contains a model $\X'= \{X'_1,\dots, X'_k\}$ of $K_k = H\setminus I$ such that for each $i$, $Z_i\subseteq X'_i$. If the answer is yes, then we get a model of $H$ by replacing the empty sets of $\X$ by the sets of $\X'$.
        If we did not find a model of $K_k$ for any choice of premodel $\X$ and subsets $Z_1,\dots, Z_k$, then we can conclude that $G$ does not contain $H$ as an induced minor.
        The algorithm described here takes polynomial time as $H$ is fixed. 
    \end{proof}

    \begin{lemma}\label{lem:splitDBB}
        The graph $S_{k,p}$ with $k\leq 3$ is $K$-NT where $K$ is the clique of $S_{k,p}$.
    \end{lemma}
    \begin{proof}  
        If $k < 3$, the degree of the vertices in the independent part of $S_{k,p}$ is at most $2$, and hence the conclusion follows from \cref{lem:minModelSmallDeg,lem:minModelPath}. 
        We consider now the case $k=3$.
        Let $H=S_{3,p}$ for some $p\geq 1$, and let $a,b,c$ be the three vertices of the clique in $H$, and $I$ the independent set of size $p$ in $H$.
        Let $G$ be a graph containing $H$ as an induced minor, and $\XH$ a model of $H$ in $G$ that minimizes the bags of $I$. 
        Suppose that there is a vertex $u$ in $I$ whose bag is non-trivial. 
        Let $P$ be a shortest path in $X_u$ from a neighbor of $X_a$ to a neighbor of $X_c$, denoted respectively by $x_a$ and $x_c$. 
        Moreover, let $P_{b}$ be a shortest path in $X_u$ from a neighbor of $X_b$ to a vertex in $P$, say $x_u$. 
        Let $P_a = x_uPx_a$ and $P_c = x_uPx_c$.
        Then the paths $P_a,P_b,P_c$, are disjoint except in $x_u$. 
        Therefore, if we replace in $\XH$ the bags $X_a$, $X_b$, $X_c$, and $X_u$ respectively by $\big (X_a \cup V(P_a)\setminus \{x_u\} \big ), \big (X_b \cup V(P_b)\setminus  \{x_u\} \big ), \big (X_c \cup V(P_c)\setminus \{x_u\} \big )$ and $\{x_u\}$, we obtain a new model of $H$ in $G$ included in $\XH$. 
        Moreover, in this model, the bag of $u$ only contains $x_u$, and thus is smaller than in $\XH$ (and the bags of the other vertices of $I$ are the same as in $\XH$), contradicting the choice of $\XH$. 
    \end{proof}

    Observe that the above result is not true if $k>3$. Indeed, there might be no vertex of $G$ of degree at least $k$ in the bags of the vertices of the independent part of $S_{k,p}$.
    Combining \cref{lem:splitDBB} and \cref{thm:DBBpoly}, we thus obtain \cref{thm:splitPoly}, restated below.
    
    \ThmCompleteSplits*
    
    \section{\texorpdfstring{\IMC{$H$}}{H-IMC} on graphs with no long induced paths} \label{sec:Pt_free}

    In this section, we show that if the input graph does not contain long induced paths, then \IMC{$H$} can be solved in polynomial time, for any fixed graph $H$.
    This allows us to develop polynomial-time algorithms for \IMC{$Gem$} and \IMC{$\FH$}. 
    In what follows, we write that a graph is $P_t$-free if it excludes the path on $t$ vertices as an induced subgraph. 
    The idea of the following result is that a minimal induced minor model in a $P_t$-free graph contains a bounded number of vertices. Hence, an exhaustive search of the possible models of $H$ can be done in polynomial time. 
    
    \ThmPtFree*
    
    \begin{proof}
    Let $G$ be a $P_t$-free graph that contains $H$ as an induced minor. 
    We may assume that $t \geq 2$; otherwise $G$ has no vertex and the problem becomes trivial.
    
    Let $\XH$ be a minimal model of $H$ in $G$. 
    We claim that the size of a bag $X_u$ of $\XH$ is bounded by $1+ \deg_H(u) \cdot (t-2)$.
    If $\deg_H(u) \leq 1$, then \cref{lem:minModelSmallDeg} implies that $|X_u|=1$, which is bounded from above by $1+ \deg_H(u) \cdot (t-2)$ since $t \geq 2$.
    Assume now that $\deg_H(u) > 1$ and let $x \in X_u$ be a vertex adjacent to some other vertex in $X_v$, for $v\in N_H(u)$.
    Let $N = N_H(u) \setminus \{v\}$, and consider a set $S\subseteq X_u$ of at most $\deg_H(u) -1$ vertices, such that, for every $w\in N$, it holds $S \cap N_G(X_w) \neq \emptyset$.
    For each $y\in S$, fix a shortest path from $x$ to $y$ in $G[X_u]$.
    Note that, since $G$ is $P_t$-free, each such path contains $x$ plus at most $t-2$ vertices.
    Let $X'_u$ be the union of the vertex sets of all these paths.
    Hence, $|X'_u| \leq 1+ \deg_H(u) \cdot (t-2)$.
    Observe that $G[X'_u]$ is connected and that, by construction, there exists an edge between $X'_u$ and $X_w$, for every $w \in N_H(u)$.
    It follows by the minimality of $X_u$ that $X_u = X'_u$, and thus we have that $|X_u| \leq 1+ \deg_H(u) \cdot (t-2)$, as claimed.
    
    From this bound on the size of minimal bags in models of $H$ in $P_t$-free graphs, we derive the following algorithm.
    We try every combination of $|V(H)|$ subsets of $V(G)$ of size at most $1+(|V(H)|-1)\cdot(t-2)$, and we either find a model of $H$ or conclude that $H\nim G$. 
    \end{proof}

    \authorcite{BKRT19} showed that the class of $H$-induced minor-free graphs are well-quasi-ordered by induced minors if and only if $H$ is an induced minor of the Gem or $\FH$. Moreover, they showed decomposition theorems for these two classes of graphs. We make use of these theorems to test \IMC{$H$} in polynomial time for the Gem and $\FH$. 
    
    \begin{theorem}[{\cite[Theorem~3]{BKRT19}}]
        \label{thm:gem}
        Let $G$ be a $2$-connected graph such that $Gem \nim G$.
        Then $G$ has a subset $X \subseteq V(G)$ of at most six vertices such that every connected component of $G \setminus X$ is either a cograph or a path whose internal vertices are of degree two in $G$. 
    \end{theorem}
    
    The idea for the algorithm is the following.
    If a graph $G$ does not have the structure of Gem-induced-minor-free graphs, then we conclude that $Gem \im G$; otherwise, we show that we can check if $Gem \im G$ in polynomial time in the restricted structure of $Gem$-induced minor-free graphs.
    
    \begin{theorem}\label{th:gemPoly}
        \IMC{$Gem$} is polynomial-time solvable.
    \end{theorem}
    
    \begin{proof}
    First, since the Gem is 2-connected, we may assume that $G$ is $2$-connected. If it is not the case, we can simply consider the $2$-connected components of $G$ independently.
    
    The algorithm is as follows: We test all subsets $X \subseteq V$ of size at most six and check whether the connected components of $G \setminus X$ meet the requirements of \cref{thm:gem}, that is, are cographs or paths whose internal vertices have degree 2 in $G$.
    Note that cographs, which are exactly $P_4$-free graphs, can be recognized in linear time~\cite{CPS85}. 
    If such a set $X$ does not exist, then $Gem\im G$.
    Hence, we may assume that the algorithm finds such a set $X$.
    We contract the internal vertices of components of $G \setminus X$ that are paths of length at least $3$ to $P_3$.
    Let $G'$ be the obtained graph and observe that $Gem \im G'$ if and only if $Gem \im G$.
    Since $|X| \leq 6$, the longest induced path in $G'$ is of length at most $26$, obtained by alternating between paths on at most $3$ vertices in $G \setminus X$ and vertices of $X$. Therefore, $G'$ is $P_{28}$-free and we can use \cref{thm:ptfree} to conclude.
    \end{proof} 
    
    We use a similar approach for the Full House (denoted $\FH$ hereafter), but the structure of $\FH$-induced minor-free graphs is more subtle, leading to more cases to consider. A \emph{wheel} is defined 
    as a graph obtained from a cycle $C$ together with an isolated vertex that is adjacent 
    to at least one vertex of $C$. A graph $G = (V,E)$ is \emph{complete multipartite} if 
    its vertex set can be partitioned into pairwise completely adjacent sets. 
    
    \begin{theorem}[{\cite[Theorem~2]{BKRT19}}]
        \label{thm:fullhouse}
        Let $G$ be a $2$-connected graph such that $\FH  \nim G$.
        Then one of the following property holds:
        \begin{enumerate}
            \item\label[property]{item:fh1} $K_4 \nim G$,
            \item\label[property]{item:fh2} $G$ is a subdivision of a graph among $K_4$, $K_{3,3}$ and the prism,
            \item\label[property]{item:fh3} $V(G)$ has a partition $(W,M)$ such that $G[W]$ is a wheel on at most $5$ vertices and $G[M]$ is a complete multipartite graph,
            \item\label[property]{item:fh4} $V(G)$ has a partition $(C,I)$ such that $G[C]$ is a cycle, $I$ is an independent set and every vertex of $I$ has the same neighborhood on $C$. 
        \end{enumerate}
    \end{theorem}
    
    \begin{theorem}\label{th:fhPoly}
        \IMC{$\FH $} is polynomial-time solvable.
    \end{theorem}
    
    \begin{proof}
    First, since $\FH$ is 2-connected, we may assume that $G$ is $2$-connected; otherwise, we may simply consider the $2$-connected components of $G$ independently. 
    Note that \cref{item:fh1,item:fh2,item:fh3,item:fh4} can all be tested in polynomial time.
    In particular, \cref{item:fh2} can be tested in polynomial time by iteratively contracting an edge incident to a vertex with degree $2$ until there is none, 
    and then checking if the resulting graph is $K_4$, $K_{3,3}$ or the prism.
    Hence, we first test, in polynomial time, whether at least one of the properties of \cref{thm:fullhouse} holds for $G$.
    If that is not the case, then $G$ must contain $\FH$ as an induced minor.
    Let us hence assume that at least one of these properties holds. 
    
    If \cref{item:fh1} holds: Then $\FH \nim G$ since a graph that does not contain $K_4$ as an induced minor cannot contain $\FH$ as an induced minor. 
    
    If \cref{item:fh2} holds: If $G$ is a subdivision of a $K_4$, then $\FH \nim G$. If $G$ is a subdivision of the prism, then $\FH \im G$ if and only if at least one edge of a triangle of the prism is subdivided. Indeed, if no edge of a triangle is subdivided, then all the vertices of $G$ must be used to obtain an induced minor model of $K_4$, and thus $\FH \nim G$. However, if an edge of a triangle is subdivided, then we have a model of $\FH$ as illustrated in \cref{fig:full_house_model}. Note that this construction can be straightforwardly extended if there is more than one subdivision.
    Similarly, if the graph is $K_{3,3}$, then $\FH \nim G$, but if it has at least one subdivided edge, then $\FH \im G$ as illustrated  in \cref{fig:full_house_model}. 
    
    If \cref{item:fh3} holds: 
    Note that an induced path in a multipartite graph is of length at most $2$. Moreover, since $|W| \leq 5$ it follows that an induced path in $G$ is of bounded length. Hence, we can use \cref{thm:ptfree} to conclude. 
    
    If \cref{item:fh4} holds: We show that $\FH  \im G$ if and only if $|I|\geq 2$ and either:
        \begin{enumerate}
            \item $I$ is not adjacent to at least one vertex of $C$, $|C|\geq 4$, and vertices in $I$ have degree at least $3$, or
            \item $I$ is adjacent to every vertex of $C$ and $|C|\geq 5$.
        \end{enumerate}
    Observe first that if the degree of the vertices in $I$ is at most $2$, then $K_4 \nim G$ as there are at most $2$ vertices of degree at least $3$ in $G$. Suppose now that the vertices in $I$ have degree at least  $3$, and observe that $K_4 \im G$. If $|I|=1$, any induced minor model of $K_4$ in $G$ contains all vertices of $G$, hence $\FH  \nim G$. Suppose now that $|I|\geq 2$. If $|C| = 3$, then observe that $\FH  \nim G$.
    If there is a vertex $u\in C$ not adjacent to $I$, then for $|C| = 4$, $|I| = 2$, we can construct a model of the $\FH $ as illustrated in \cref{fig:full_house_model}. 
    This extends straightforwardly for $|C| \geq 4$ and $|I| \geq 2$.
    
    Suppose now that $I$ is adjacent to every vertex of $C$. 
    Observe that at least $3$ bags of a model of $K_4$ must contain vertices of $C$. Hence, the bag of the degree-$2$ vertex of $\FH$ must not contain vertices of $I$. Therefore, if $|C| = 4$ all the vertices of $C$ are in different bags, and one bag must contain one vertex of $I$. In such model, each is adjacent to at least $3$ other bags, we can conclude that if $|C| = 4$ and $|I|\geq 2$, then $\FH \nim G$.
    If $|C| = 5$ and $|I|=2$, then we can construct a model of the $\FH $ as illustrated in \cref{fig:full_house_model}. 
    This extends straightforwardly for $|C| \geq 5$ and $|I| \geq 2$.
    \end{proof}

    \begin{figure}
        \centering
        \includegraphics[width=0.9\linewidth]{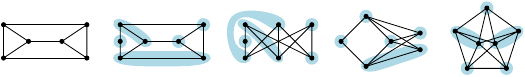}
        \caption{From left to right: the prism; models of $\FH$ in a subdivided prism and a subdivided $K_{3,3}$; models of $\FH$ for graphs with the \cref{item:fh4} in \cref{thm:fullhouse}. }
        \label{fig:full_house_model}
    \end{figure}

    \printbibliography
    
    \end{document}